\documentclass[dvips,aoas,authoryear]{imsart}

\usepackage{amsthm}
\usepackage{amsmath}
\usepackage{amsfonts}
\usepackage{latexsym}
\usepackage{amssymb}
\usepackage{graphicx}
\usepackage{epsfig}
\usepackage{natbib}

\RequirePackage[colorlinks,citecolor=blue,urlcolor=blue]{hyperref}
\RequirePackage{hypernat}

\startlocaldefs
\theoremstyle{plain}
\newtheorem{thm}{Theorem}

\theoremstyle{definition}

\theoremstyle{remark}

\endlocaldefs

\begin{document}

\begin{frontmatter}

\title{Manifold embedding for  curve registration}
\runtitle{Manifold embedding for  curve registration}

\author{\fnms{Chlo\'e} \snm{Dimeglio}\ead[label=e3]{cd@geosys.com}}
\address{Institut de Math\'ematiques de Toulouse\\
\printead{e3}}
\affiliation{Institut de Math\'ematiques de Toulouse \& Geosys}
\and
\author{\fnms{Jean-Michel} \snm{Loubes}\ead[label=e1]{Jean-Michel.Loubes@math.univ-toulouse.fr}}
\address{Institut de Math\'ematiques de Toulouse\\
\printead{e1}}
\affiliation{Institut de Math\'ematiques de Toulouse}
\and
\author{\fnms{Elie} \snm{Maza}\corref{k}\ead[label=e2]{Elie.Maza@ensat.fr}}
\address{Ecole Nationale Sup\'erieure Agronomique de Toulouse\\
Genomic \& Biotechnology of the Fruit Laboratory\\
UMR 990 INRA/INP-ENSAT\\
\printead{e2}}
\affiliation{Institut National Polytechnique de Toulouse}

\runauthor{C. Dimeglio, J-M. Loubes \& E. Maza}

\begin{abstract}
We focus on the problem of finding a good representative of a sample of random curves warped from a common pattern $f$. We first prove that such a problem can be moved onto a manifold framework. Then, we propose an estimation of the common pattern $f$ based on an approximated geodesic distance on a suitable manifold. We then compare the proposed method to more classical methods.
\end{abstract}



\end{frontmatter}

\section{Introduction}\label{s:intro}
The  outcome of a statistical process  is often  a sample of curves $\{f_i,\:i=1,\ldots,m\}$ showing an unknown common structural pattern, $f$, which characterizes the behaviour of the observations. Examples are numerous, among others, growth curves analysis in biology and medicine, quantitative analysis of microarrays in molecular biology and genetics, speech signals recognition in engineering, study of expenditure and income curves in economics\dots. Hence, among the last decades,  there has been a growing  interest to develop statistical methodologies which enables to recover from the observation functions  a  single "mean curve" that  conveys all the information of the data.

A major difficulty comes from the fact that there are both amplitude variation (in the $y$-axis) or phase variation (in the $x$-axis) which prevent any direction extraction of the mean, median, correlations or any other statistical indices for a standard multivariate procedure such as principal component analysis, and canonical correlations analysis, see \cite*{Kneip1992}  or \cite*{Ramsay2005} and references therein. Indeed, the classical cross-sectional mean does not provide a consistent estimate of the function of interest $f$ since it fails to capture the  structural characteristics in the sample of curves as quoted in \cite*{Ramsay1998}. Hence, curve registration methods (also called curve alignment, structural averaging, or time warping) have been proposed in the statistical literature. We refer to, just to name a few, \cite*{Sakoe1978} in Spoken Word Recognition domain, \cite{Kneip1992} for Landmark Registration, \cite*{Silverman1995} for a functional principal component analysis, \cite*{Wang1997} for Dynamic Time Warping, \cite{Ramsay1998} for Continuous Monotone Registration, \cite*{Ronn2001} for shifted curves, \cite*{Liu2004} for functional convex averaging, \cite*{Gervini2005} for maximum likelihood estimation, \cite*{Gamboa2007} for shifts estimation, \cite*{James2007} for alignment by moments, and \cite*{Dupuy2011} for Structural Expectation estimation.

This issue is closely related to the problem of finding the mean of observations lying in a space with an unknown, non necessarily euclidean, underlying geometry. The problem is thus twofold.

First, the mere definition of the mean should be carefully studied. Indeed, let $\mathcal{E}=\left\{X_1,\dots,X_n\right\}$ be a sample of i.i.d random variables of law $X\in\mathcal{M}$ where $\mathcal{M}$ is a submanifold of $\mathbb{R}^p$.   If we denote by $\mathrm{d}$ the Euclidean distance on $\mathbb{R}^p$, then  the classical sample mean, or Fr\'echet sample mean, defined by
\begin{equation}\label{classicalMean}
\widehat{\mu}=\arg\min_{\mu\in\mathbb{R}^p}\sum_{i=1}^n\mathrm{d}^2\left(X_i,\mu\right)
\end{equation}
is not always a good representative of the given sample $\mathcal{E}$, and, obviously, of the underlying population. Using the geometry of the manifold, it  seems natural to replace Criterion~\eqref{classicalMean} by
\begin{equation*}
\widehat{\mu}_I=\arg\min_{\mu\in\mathcal{M}}\sum_{i=1}^n\delta^2\left(X_i,\mu\right)
\end{equation*}
where $\delta$ is the geodesic distance on manifold $\mathcal{M}$, giving rise to the \textit{intrinsic mean}, whose existence and properties are studied, for instance, in \cite*{Bhattacharya2003}.  When dealing with functional data, we assume that the functions $f_i$ can be modeled  as variables with values on a manifold, and curve registration amounts to considering an intrinsic statistic that reflects the behaviour of the data. In the following we will consider, for $\alpha>0$,
\begin{equation}\label{intrinsicStatistic0}
\widehat{\mu}_I^\alpha=\arg\min_{\mu\in\mathcal{M}}\sum_{i=1}^n\delta^\alpha\left(X_i,\mu\right).
\end{equation}
In particular, for $\alpha=1$, we will deal with $\widehat{\mu}_I^1$,  the intrinsic sample median.

Second, previous construction relies on the choice of the embedding which may not be unique, then the manifold itself and its underlying geodesic distance. Actually we only have at hand a sample of random variables which are sought to be a discretization of an unobserved manifold. Over the last decade, some new technics have been developed to find and compute the natural embedding of data onto a manifold and to estimate the corresponding geodesic distance, see for instance \cite*{Silva2003} for a review of global (Isomap type) and local (LLE type) procedures, while applications have been widely developed, see for instance \cite*{Pennec2006}.

In the following, we will consider an   approximation, achieved with a graph theory approach inspired by works on manifold learning and dimension reduction \citep*{Tenenbaum2000}. We will first show that curve registration for parametric transformations can be solved using a manifold geodesic approximation procedure. Then, we will highlight that this enables to recover a mean pattern which conveys the information of a group of  curves. This pattern is used for curve classification for simulated data and real data which consists in predicting a particular landscape using the reflectance of the vegetation.

This article falls into the following parts. Section~\ref{s:algo} is devoted to the construction of the approximated geodesic distance. In Section~\ref{s:curve}, we describe the manifold framework point of view for curve registration. We then explain how to estimate a representative of a sample of warped curves. The performance of this estimator is then studied in Section~\ref{s:simul} using simulated data, and in Section~\ref{s:real} with a real data set. Concluding remarks are given in Section~\ref{s:conclu}. Proofs are gathered in Section~\ref{s:append}.

\section{A graph construction for topology estimation}\label{s:algo}
Let $X$ be a random variable with values in an unknown connected and geodesically complete Riemannian manifold $\mathcal{M}\subset\mathbb{R}^p$. We observe an i.i.d sample $\mathcal{E}=\{X_i\in\mathcal{M},\:i=1,\dots,n\}$ with distribution $X$. Set $\mathrm{d}$ the Euclidean distance on $\mathbb{R}^p$ and $\delta$ the  induced geodesic distance on $\mathcal{M}$. Our aim is to estimate intrinsic statistics defined by Equation~\eqref{intrinsicStatistic0}. Since the manifold $\mathcal{M}$ is unknown, the main issue is to estimate the geodesic distance between two points on  the manifold, that is $\delta\left(X_i,X_j\right)$.

Let $\gamma_{ij}$ be the geodesic path connecting two points $X_i$ and $X_j$, that is the minimum length path on $\mathcal{M}$ between points $X_i$ and $X_j$. Denoting by $\mathrm{L}\left(\gamma\right)$ the length of a given path $\gamma$ on $\mathcal{M}$, we have that $\delta\left(X_i,X_j\right)=\mathrm{L}\left(\gamma_{ij}\right)$.

In the Isomap algorithm, \cite{Tenenbaum2000} propose to learn manifold topology from a graph connecting $k$-nearest neighbors for a given integer $k$. In the same way, our purpose is to approximate the geodesic distance $\delta$ with a suitable graph connecting nearest neighbors. Our approximation is carried out in three steps. Thereafter, we denote $g_{ij}$ a path connecting two points $X_i$ and $X_j$ on a given graph, and $L\left(g_{ij}\right)$ the length of such a path.

\paragraph{Step 1} Consider $K=\left(\mathcal{E},E\right)$ the complete Euclidean graph associated to sample $\mathcal{E}$, that is the graph made with all the points of the sample $\mathcal{E}$ as vertices, and with edges
$$E=\left\{\left\{X_i,X_j\right\},\:i=1,\dots,n-1,\:j=i+1,\dots,n\right\}.$$
For an Euclidean graph, the edge weights are the edge lengths, that is, the Euclidean distances between each pair of points.

\paragraph{Step 2} Let $T=\left(\mathcal{E},E_T\right)$ be the Euclidean Minimum Spanning Tree (EMST) associated to $K$, that is, the spanning tree that minimizes
\begin{equation*}
\sum_{\{X_i,X_j\}\in E_T}\mathrm{d}\left(X_i,X_j\right).
\end{equation*}
The underlying idea in this construction is that,  if two points $X_i$ and $X_j$ are relatively close, then we have that $\delta\left(X_i,X_j\right)\approx\mathrm{d}\left(X_i,X_j\right)$. This may not be true if the manifold is very twisted and if too few points are observed, and may induce bad approximations, hence the algorithm will produce a good approximation for relatively regular manifolds. It also generally requires a large number of sampling points on the manifold in order to guarantee the quality of this approximation. This drawback is well known when dealing with graph based approximation of the geodesic distance. Then, the graph $T$ is a connected graph spanning $K$ which mimics  the manifold $\mathcal{M}$. Furthermore, an approximation of the geodesic distance $\delta\left(X_i,X_j\right)$ is provided by the sum of all the euclidean distance of the edges of the shortest path on $T$ connecting $X_i$ to $X_j$, namely $$\hat{\delta}\left(X_i,X_j\right)=\min_{g_{ij}\in T}L\left(g_{ij}\right).$$
However, this approximation is too sensitive to perturbations of the data, and hence, very unstable. To cope with this problem, we propose to add more edges between the data to add extra paths in the data sample and thus to increase stability of the estimator.  The idea is that paths which are close to the ones selected in the construction of the EMST could provide alternate ways of connecting the edges. Close should be here understood as lying in balls around the observed points. Hence, these new paths between the data  are admissible and should be added to the edges of the graph. This provides redundant information but also stabilizes the constructed distance, and may also provide an answer to the the main defect of the algorithm that considers that two points very close with respect to the Euclidean distance are also close with respect to the geodesic distance.

\paragraph{Step 3} Let $B\left(X_i,\epsilon_i\right)\subset\mathbb{R}^p$ the open ball of center $X_i$ and radius $\epsilon_i$ defined by
$$\epsilon_i=\max_{\left\{X_i,X_j\right\}\in E_T}\mathrm{d}\left(X_i,X_j\right).$$
Let graph $K^\prime=\left(\mathcal{E},E^\prime\right)$ defined by
$$\{X_i,X_j\}\in E^\prime\Longleftrightarrow\overline{X_iX_j}\subset\bigcup_{i=1}^nB\left(X_i,\epsilon_i\right)$$
where
$$\overline{X_iX_j}=\left\{X\in\mathbb{R}^p,\:\exists\lambda\in[0,1],\:X=\lambda X_j+(1-\lambda)X_i\right\}.$$
Then, $K^\prime$ is the graph which gives rise to our estimator of the distance $\delta$ :
\begin{equation}\label{estidelta}
\hat{\delta}\left(X_i,X_j\right)=\min_{g_{ij}\in K^\prime}L\left(g_{ij}\right).
\end{equation}
Hence, $\hat{\delta}$ is the distance associated with $K^\prime$, that is, for each pair of points $X_i$ and $X_j$, we have $\hat{\delta}\left(X_i,X_j\right)=\mathrm{L}\left(\hat{\gamma}_{ij}\right)$ where $\hat{\gamma}_{ij}$ is the minimum length path between $X_i$ and $X_j$ associated to $K^\prime$.

We note that, the 3-steps procedure described above contains widespread graph-based methods to achieve our purpose. In this article, our graph-based calculations, such as MST estimation or shortest path calculus, were carried out with the R Language \citep*{R2010} with the \textit{igraph} package for network analysis \citep*{Csardi2006}.

An example of this 3-steps procedure and its behaviour when the number of observations increases are displayed respectively in Figure~\ref{met1} and Figure~\ref{met2}. In Figure~\ref{met1}, points $\left(X_i^1,X_i^2\right)_i$ are simulated as follows :
\begin{equation}\label{sim1}
X_i^1=\frac{2i-n-1}{n-1}+\epsilon_i^1\mbox{ and }X_i^2=2\left(\frac{2i-n-1}{n-1}\right)^2+\epsilon_i^2
\end{equation}
where $\epsilon_i^1$ and $\epsilon_i^2$ are normaly distributed with mean 0 and variance 0.01. In Figure~\ref{met2}, we give some results of graph $K^\prime$ for $n\in\{10,30,100,300\}$. We can see in such a figure that graph $K^\prime$ tends to be close to the manifold $\left\{\left(t,t^2\right)\in\mathbb{R}^2,\:t\in\mathbb{R}\right\}$.
\begin{figure}[htb]
	\centering
	\includegraphics[width=0.7\textwidth,angle=270]{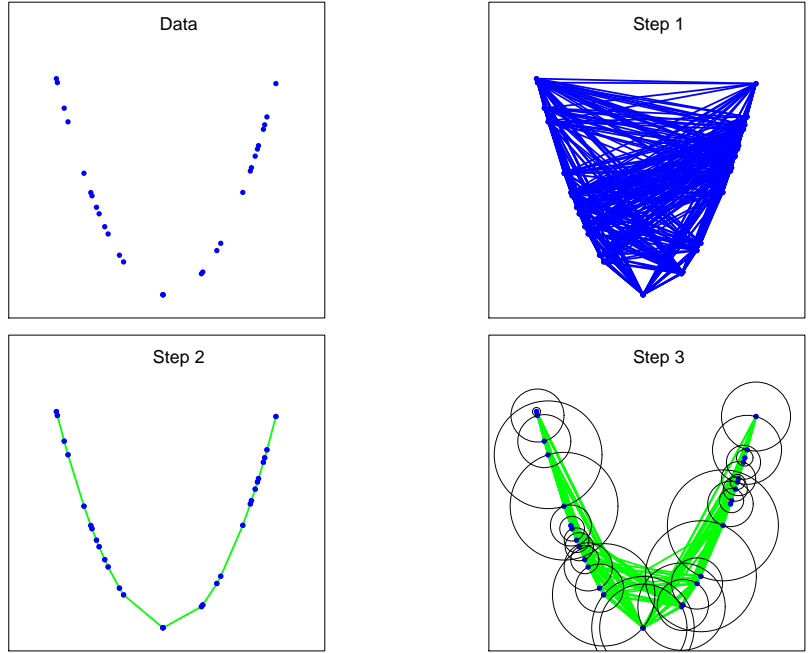}
	\caption{Construction of a subgraph $K^\prime$ from Simulation~\eqref{sim1} with $n=30$ points. On the top left, a simulated data set. On the top right, the associated complete Euclidean graph $K$ (Step 1). On the bottom left, the EMST associated with the complete graph $K$ (Step 2). On the bottom right, the associated open balls and the corresponding subgraph $K^\prime$ (Step 3).}
	\label{met1}
\end{figure}
\begin{figure}[htb]
	\centering
	\includegraphics[width=0.7\textwidth,angle=270]{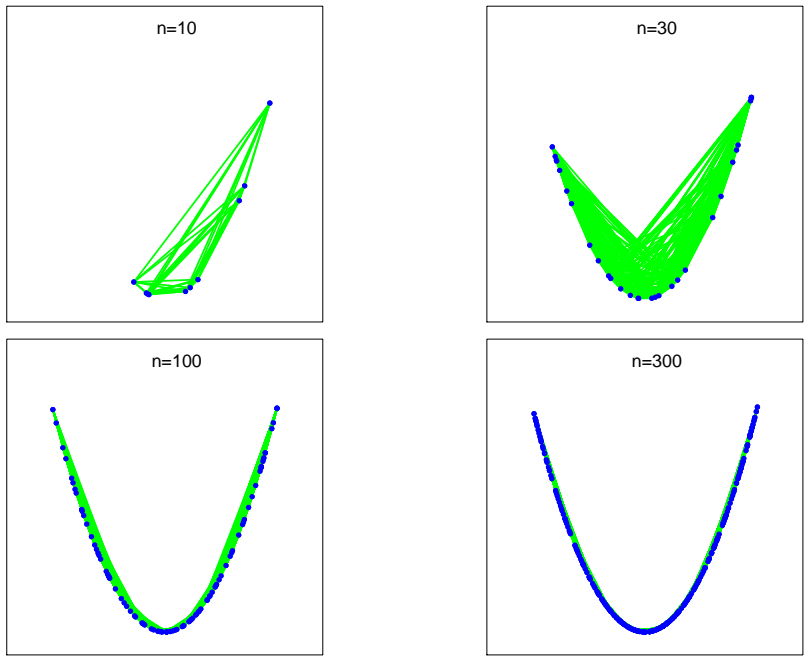}
	\caption{Evolution of graph $K^\prime$ for Simulation~\eqref{sim1} and $n\in\{10,30,100,300\}$.}
	\label{met2}
\end{figure}

The main difference between our algorithm and the Isomap algorithm lies in the treatment of points which are far from the others. Indeed, the first step of the original Isomap algorithm consists in constructing the $k$-nearest neighbor graph or the $\epsilon$-nearest neighbor graph for a given integer $k$ or a real $\epsilon>0$. Hence, points which are not connected to the biggest graph, since they are too distant, are not used for the construction of the estimated distance. Such a step is not present in our algorithm since in the applications we consider a distant point is not always an outlier. Hence, we do not exclude any points, and rather, for the construction of the EMST, all points of the data set are  connected.  Moreover, the Isomap algorithm requires the choice of parameters which are closely related to the local curvature of the manifold (see, for instance, \cite*{Balasubramanian2002}). This involves a heavy computing phase which is crucial for the quality of the construction, while, in our version we tend to give an automatic selection of parameters. We will show in Section~\ref{s:curve} that both procedures used for curve registration behave in a similar way and over performs other standard feature extraction methods.

In the following section, we present a new application of manifold learning to the curve alignment problem.

\section{Application to curve alignment}\label{s:curve}
Consider a function $f:\mathbb{R}\to\mathbb{R}$, which will be the  pattern to be recovered, observed in a translation effect framework. Let $A$ be a real valued random variable with unknown distribution on an interval $(b,c)\subset\mathbb{R}$. The observation model is defined by
\begin{equation}\label{model_shift}
X_i^j=f\left(t_j-A_i\right),\:i\in\{1,\dots,n\},\:j\in\{1,\dots,m\},
\end{equation}
where $\left(A_i\right)_i$ are i.i.d random variables drawn with distribution  $A$ which model the unknown translation  parameters, while  $\left(t_j\right)_j\in\mathbb{R}^m$ stand for the measurement points.

This situation usually happens  when  individuals  experience similar events, which are explained by a common pattern $f$, and when the starting times of the events are not synchronized. Such a model has been studied, for instance, in \cite{Silverman1995} and in \cite{Ronn2001}. This issue has also received a specific attention in a semi-parametric framework in \cite{Gamboa2007} or \cite*{Castillo2009}. In these works, among others, shift parameters are estimated, which enables to align the observations and thus to get rid of the translation issue. Model~\eqref{model_shift} falls also under the generic warping model proposed in \cite{Maza2006} and in \cite{Dupuy2011} which purpose is to estimate the underlying {\it structure} of the curves. For this, the authors define  the \textit{structural median} $f_\mathrm{SM}$ of the data. In the case of translation effects, it corresponds to
\begin{equation}\label{def_structuralMedian}
f_\mathrm{SM}=f\left(\cdot-\mathrm{med}(A)\right)
\end{equation}
with $\mathrm{med}(A)$ the median of $A$. Hence, a natural estimator of the structural median $f_\mathrm{SM}$, related to Model~\eqref{model_shift}, would be
\begin{equation}\label{est_structuralMedian}
\widehat{f}_\mathrm{SM}=\left(f\left(t_1-\widehat{\mathrm{med}}(A)\right),f\left(t_2-\widehat{\mathrm{med}}(A)\right),\dots,f\left(t_m-\widehat{\mathrm{med}}(A)\right)\right)
\end{equation}
with $\widehat{\mathrm{med}}(A)$ the median of sample $\left(A_i\right)_i$. However, we first note that the translation parameters $\left(A_i\right)_i$ are not observed, and, as a consequence, that the median $\widehat{\mathrm{med}}(A)$ can not directly be calculated. Then, the  function $f$ is also unknown, so, estimating $\widehat{\mathrm{med}}(A)$ is not enough to calculate $\widehat{f}_\mathrm{SM}$.  Our purpose is to show that our manifold point of view provides a direct estimate of  ${f}_\mathrm{SM}$ without the prior estimation of $\mathrm{med}(A)$.

In order to use the manifold embedding approach, define
$$\begin{array}{rcl}
X:\mathbb{R}&\to&\mathbb{R}^m\\
a&\mapsto&X(a)=\left(f\left(t_1-a\right),f\left(t_2-a\right),\dots,f\left(t_m-a\right)\right)
\end{array}$$ 
and set \begin{equation*}
\mathcal{C}=\left\{X(a)\in\mathbb{R}^m,\:a\in\mathbb{R}\right\}.
\end{equation*}
As soon as $f'\neq0$, the map $X:a\mapsto X(a)$ provides a natural parametrization of $\mathcal{C}$ which can thus be seen as  a submanifold of $\mathbb{R}^m$ of dimension 1. The corresponding geodesic distance is given by $$\delta\left(X(a_1),X(a_2)\right)=\left|\int_{a_1}^{a_2}\left\|X^\prime(a)\right\|\mathrm{d}a\right|.$$
The observation model~\eqref{model_shift} can be seen as a discretization of the manifold $\mathcal{C}$ for different values $\left(A_i\right)_i$. Finding the median of all the shifted curves can hence be done by understanding the \textit{geometry} of space $\mathcal{C}$, and thus approximating the geodesic distance between the curves.

The following theorem states that the structural median $\widehat{f}_\mathrm{SM}$ defined by Equation~\eqref{est_structuralMedian} is equivalent to the median with respect to the geodesic distance on $\mathcal{C}$, that is
$$\widehat{\mu}_I^1=\arg\min_{\mu\in\mathcal{C}}\sum_{i=1}^n\delta\left(X_i,\mu\right),$$
which provides a geometrical interpretation of the structural median.
\begin{thm}\label{thmSM}
Under the assumption that $f' \neq 0$, we get that
$$\widehat{\mu}_I^1=\widehat{f}_\mathrm{SM}.$$
\end{thm}
Previous theorem can be extended to the more complex case of parametric deformations of the type
\begin{eqnarray*}
X:\mathbb{R}^3&\to&\mathbb{R}^m\\
(a,b,c)&\mapsto&X(a,b,c)=\left(af\left(t_1-b\right)+c,\dots,af\left(t_m-b\right)+c\right)
\end{eqnarray*}
as soon as $a\neq0$ and $f'\neq0$. Such a model has been described in \cite*{Vimond2010} and in \cite*{Bigot2010b}. In this case, the submanifold is obviously of dimension 3.

In an even more general framework, when the observations can be modeled by  a set of curves warped one from another by an unobservable deformation process, this estimate enables to recover the main pattern. It relies on the assumption that all the data belong to a manifold whose geodesic distance can be well approximated by the graph structure of the modified minimal spanning tree described in Section~\ref{s:algo}.

Finally, we propose the following estimator of the structural median 
\begin{equation}\label{estimator}
\tilde{\mu}_I^1=\arg\min_{\mu\in\mathcal{E}}\sum_{i=1}^n\widehat{\delta}\left(X_i,\mu\right),
\end{equation}
using the geodesic distance $\widehat{\delta}$, estimated by the algorithm described in Section~\ref{s:algo}.

The numerical properties of this estimator is studied using simulations in Section~\ref{s:simul}, and for real data sets in Section~\ref{s:real}.

\section{Simulations}\label{s:simul}
We consider the target function $f:\mathbb{R}\to\mathbb{R}$ defined by $f(t)=t\sin(t)$. We simulate deformations  of this function on $j=1,\dots,m=100$ equally distributed points $t_j$ of the interval $[-10,10]$, according to the following model :
\begin{equation}\label{sim2}
Y_i\left(t_j\right)=A_if\left(B_it_j-C_i\right),\:i=,\dots,n,\:j=1,\dots,m,
\end{equation}
where $\left(A_i\right)_i$  and $\left(C_i\right)_i$ are i.i.d uniform random variables on $[-10,10]$ while $\left(B_i\right)_i$ is an i.i.d sample of a uniform distribution on $[-1,1]$. We finally obtain a data set of $n=100$ curves where each differs from the initial function $f$ by a translation and an amplitude deformation. The data is displayed on the left graph of Figure~\ref{curve_estim}.

We then consider four estimators of the function $f$. The first one, which minimizes the approximated geodesic distance, defined by Equation~\eqref{estimator}, will be referred to as the structural median estimator. The second one is obtained by the Curve Alignment by Moments procedure (CAM) developped by \cite*{James2007}. The third one is the template obtained with the Isomap strategy, with the "isomap" function of the R package \textit{vegan} \citep{Oksanen2011}. The last one is the mere mean of the data.

We recall here that the CAM procedure consists on extracting the mean pattern by synchronization of the moments of the simulated curves. For this, \cite{James2007} introduces the {\it feature function} concept for a given function $g$, defined as $I_g(t)$ :
\begin{equation*}
I_g(t)\geq0\mbox{ and }\int{I_g(t)\mathrm{d}t}=1
\end{equation*}
and the moments 
\begin{equation*}
\mu_g^{(1)}=\int{tI_g(t)\mathrm{d}t}\mbox{ and }\mu_g^{(k)}=\int{\left(t-\mu_g^{(1)}\right)^kI_g(t)\mathrm{d}t},\:k\geq2.
\end{equation*}
Then, the CAM procedure align the curves by warping their moments, for instance, the amplitude of the peaks, at the location they occur, the variance around these peaks, and so on. This method relies on the choice of a proper feature function, for instance $I^{(l)}_g(t)=|g^{(l)}(t)|/\int|g^{(l)}(s)|\mathrm{d}s $ for a given $l \geq 0$, on an approximation of the functions by splines, and the selection of the number of moments to be synchronized. Hence, it highly depends on the choice of  these tuning parameters. We have chosen the optimal value of the parameters over a grid.

These four estimators are shown on Figure~\ref{curve_estim}. With the CAM or the mere mean procedure, the average curve does not reflect the structure of the initial curves, or the amplitude of their variations. On the contrary, the structural median extracted by Manifold Warping has the characteristics of the closest target curve, but is also its best approximation together with the pattern obtained with the Isomap strategy. Note here that our version of the algorithm for curves provide the same kind of template and is parameter free while parameters governing the dimension of the manifold embedding must be chosen for the Isomap procedure. Nevertheless, both procedures are competitive and lead to similar performance.  
\begin{figure}[htb]
	\centering
	\includegraphics[width=0.7\textwidth,angle=270]{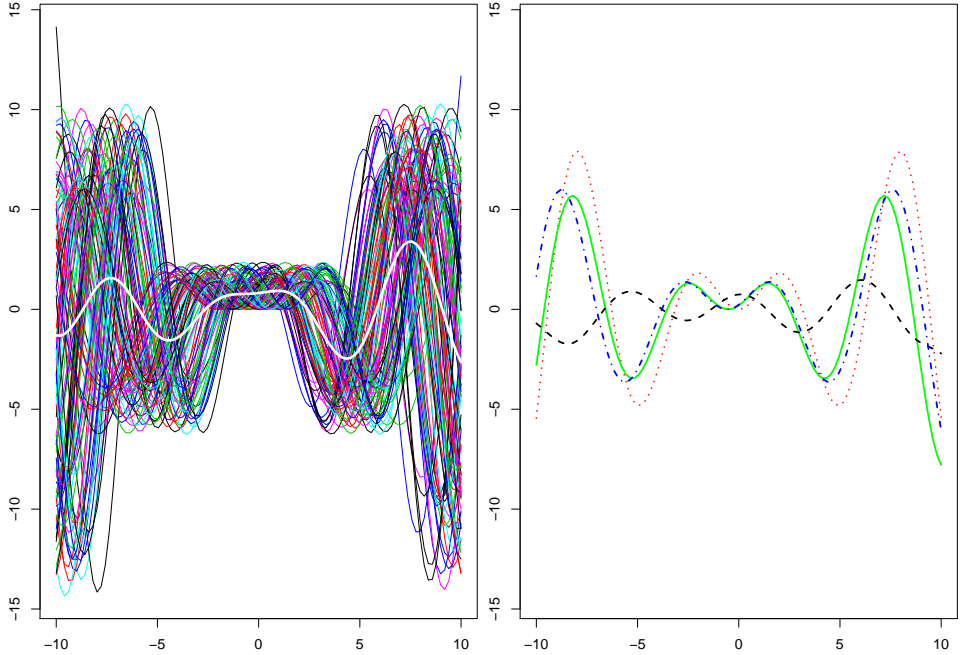}
	\caption{On the left, a simulated data set of warped curves from Model~\eqref{sim2} and an estimation of $f$ with the mere mean (white line). On the right, the target function $f$ (red dotted line), an estimation of the structural median by Manifold Warping (green solid line), an estimation obtained by Isomap (blue dot-dashed line), and an estimation obtained with the CAM procedure (black dashed line).}
	\label{curve_estim}
\end{figure}
\section{Real data}\label{s:real}
Consider the real data case where an observation curve represents the reflectance of a particular landscape and fully characterizes the nature of each landscape. The purpose of this study is to predict the different landscapes while observing the reflectance profiles. In  Figures~\ref{real_data_ble} and \ref{real_data_mais}, we present two data sets corresponding to reflectance patterns of two landscapes in the same region with the same period. However, the reflectance depends on the vegetation whose growth depends on the weather condition and the behavior in soil. It is therefore relevant to consider that these profiles are deformations in translation and/or amplitude of a single representative function of the reflectance behaviour of each landscape in this region at this time.
\begin{figure}[htb]
	\centering
	\includegraphics[width=0.7\textwidth,angle=270]{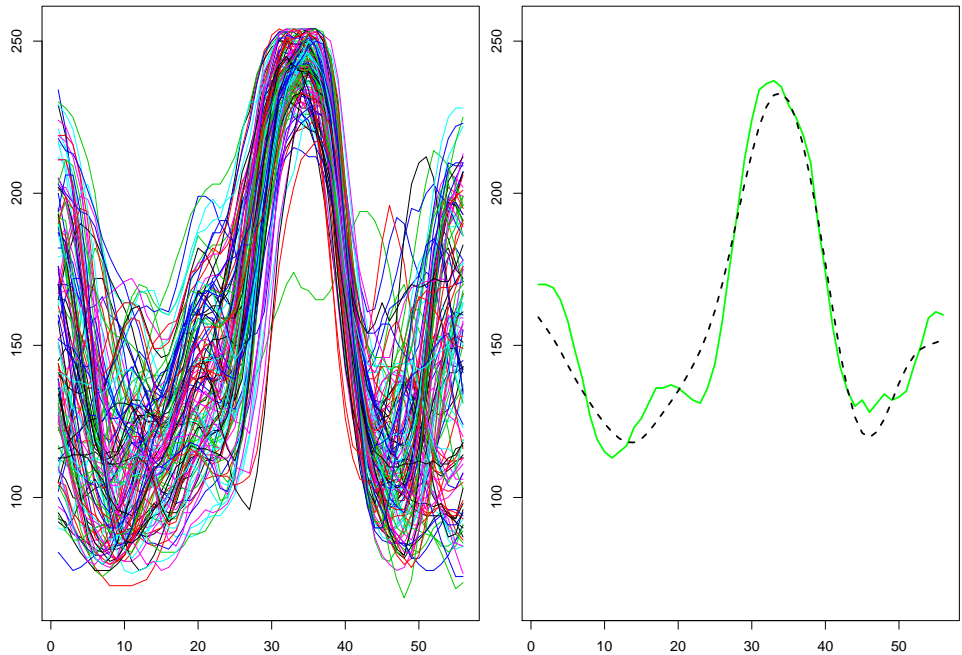}
	\caption{On the left, the first landscape data. On the right, the CAM representative estimation (black dashed line) and the Manifod Warping estimation (green solid line).}
	\label{real_data_ble}
\end{figure}
\begin{figure}[htb]
	\centering
	\includegraphics[width=0.7\textwidth,angle=270]{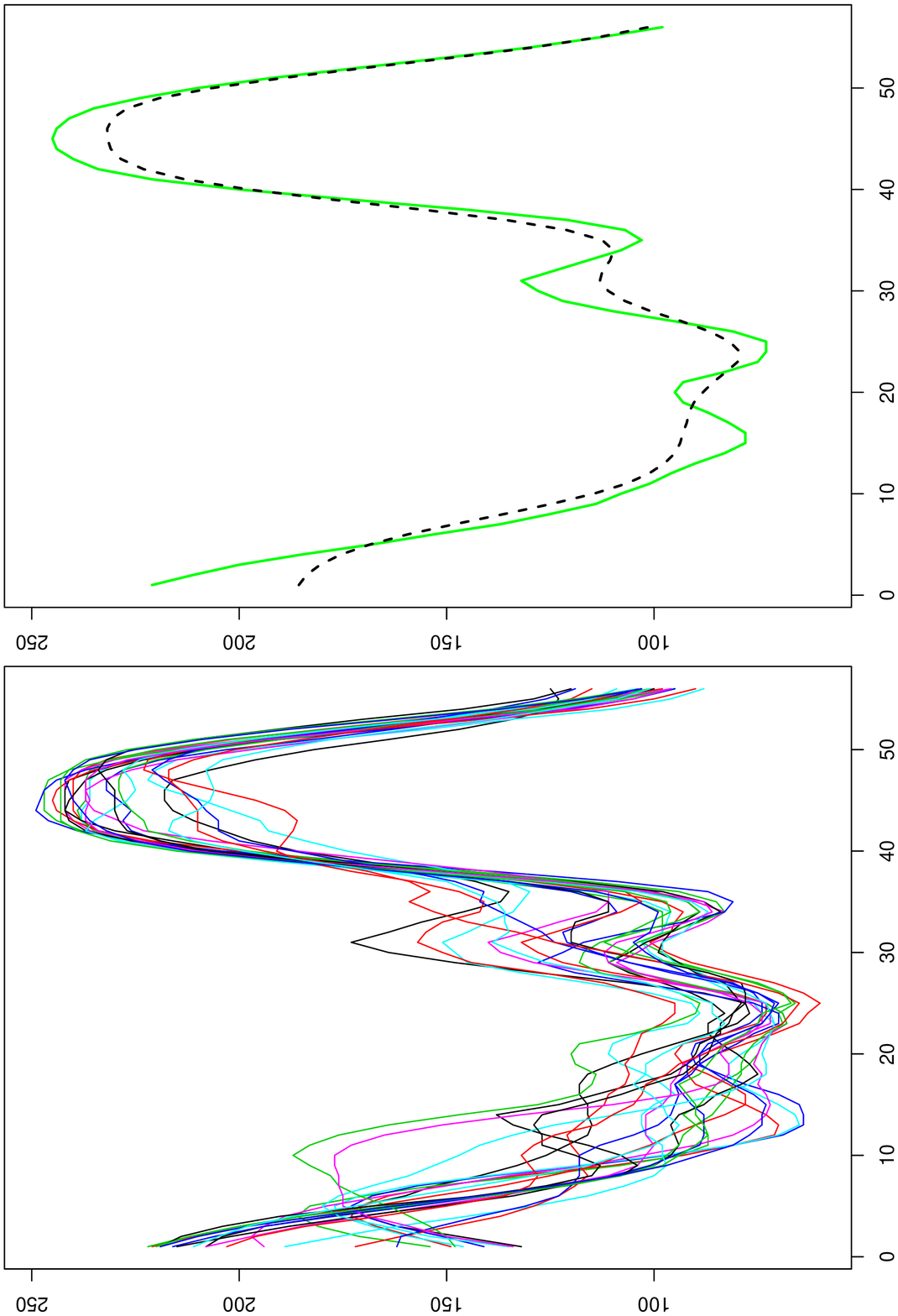}
	\caption{On the left, the second landscape data. On the right, the CAM representative estimation (black dashed line) and the Manifod Warping estimation (green solid line).}
	\label{real_data_mais}
\end{figure}

Our aim is to build a classification procedure. For this, we will use  a labeled set of curves and extract from each group of similar landscape  a representative profile. Then, we will allocate a new curve to the group whose representative curve will be the closest. That is the reason why it is  important to obtain a pattern which captures the structure of the curves. We will use three different ways to get a representative group of curves, the mean curve, the CAM method and our method, referred to as the Manifold Warping. We will compare their classification performance  together with a usual classification procedure : the classical $k$-nearest neighbours.

In Figure~\ref{real_data_ble}, we observe that the CAM average oversmoothes the peaks of activity at times 12 and 22 to make them almost nonexistent. This is a clear defect since, according to the experts of landscape remote sensing, these peaks of activity are representative of the nature of landscape. Indeed, these peaks convey essential informations which determines, among other things, the type of landscape. On the other hand, these changes are very well rendered by the pattern obtained  by Manifold Warping. The same conclusions can be drawn in Figure~\ref{real_data_mais} for an other landscape. In this application domain, extracting a curve by Manifold Warping is best able to report data as reflecting their structure and thus  to obtain a better representative.

Now,  we try to identify "unknown" landscapes by comparing each curve to the mean pattern of each group. The allocation rule is built using the Euclidean distance. Note that here we have sought to classify the landscapes, not using the whole curve which would correspond to a whole year of observation but using only a part of the curves, namely all the observations before $t=30$. To benchmark our procedure,  we compare our performance to  the method of the $k$-nearest neighbors classification.  Finally, we obtain the confusion matrices displayed in Tables~\ref{GJ_mat}  and~\ref{knn_mat}. We get a much better discrimination of landscapes with the method consisting in estimating a representative by Manifold Warping than by the CAM method or by classical mean.
\begin{table}[htb]
	\begin{tabular}{|c|c|c||c|c|}
	\hline
 	\textbf{Pixel} & \multicolumn{2}{|c||}{\textbf{Manifold classification}} & \multicolumn{2}{|c|}{\textbf{CAM classification}}\\
 	\hline
	\textbf{Reference} & Landscape1 & Landscape2 & Landscape1 & Landscape2\\
	\hline
	Landscape1 & 21 & 0 & 12 & 9\\
	\hline
	Landscape2 & 1 & 19 & 1 & 19\\
	\hline
	\end{tabular}
  \caption{Manifold Warping and CAM confusion matrices.}
  \label{GJ_mat}
\end{table}
\begin{table}[htb]
	\begin{tabular}{|c|c|c||c|c|}
	\hline
 	\textbf{Pixel} & \multicolumn{2}{|c||}{\textbf{Mean classification}} & \multicolumn{2}{|c|}{\textbf{$k$-nn classification}}\\
 	\hline
	\textbf{Reference} & Landscape1 & Landscape2 & Landscape1 & Landscape2\\
	\hline
	Landscape1 & 12 & 9 & 15 & 6\\
	\hline
	Landscape2 & 0 & 20 & 2 & 18\\
	\hline
	\end{tabular}
  \caption{Classical mean and $k$-nearest neighbors confusion matrices.}
  \label{knn_mat}
\end{table}
\section{Conclusion}\label{s:conclu}
By using an Isomap inspired strategy, we have extracted from a pattern of curves, a curve which, playing the role of the mean, serves as a pattern conveying the information of the data. In some cases, in particular when the structure of the deformations entails that the curve can be embedded into a manifold regular enough, we have shown that this corresponds to finding the structural expectation of the data, developed in \cite{Dupuy2011}, which improves the performance of other {\it mean} extraction methods. This enables to derive a classification strategy that assigns a curve to the group, whose representative curve is the closest, with respect to the chosen distance. Of course, the performance of this allocation rule deeply relies on the good performance of the pattern extraction.

One of the major drawbacks of this methodology are that first  a high number of data are required in order to guarantee a good approximation of the geodesic distance at the core of this work. Actually, note that the number of observations, i.e the sampling rate of the manifold highly depends on the regularity of the manifold such that the assumption that the euclidean path between two observations follow approximatively the geodesic path. Hence, the data set should be carefully chosen for the manifold to be smooth enough. We point out that an enhancement could come from a prior registration procedure  first applied to the curve and then the manifold warping procedure applied to the registered data.

The second drawback which may  also be viewed as an advantage, is the following : the extracted pattern is a curve that belong to the observations. One the one hand, it may contains noise if the data are noisy observations, but on the other hand it thus guarantees that the pattern shares the mean properties and specifies of the observations. A solution when the noise must be removed is either to directly smooth the resulting pattern or to consider the neigbourhood of the extracted pattern with respect to the approximated geodesic distance and then use a kernel estimator with these observations to obtain a regularized {\it mean} curve.

Nevertheless, we promote this procedure when a large amount of data are available and when the sets of similar curves share a common behaviour which fully characterizes the observations, coming from an economic, physical or biological model for instance. This methods has been applied with success to a large amount of cases. Numerical packages for R or Matlab are available on request.
\section{Appendix}\label{s:append}
\begin{proof}[Proof of Theorem \ref{thmSM}]
Take  $\mu=X(\alpha)$ with $\alpha\in]b,c[$, we can write
\begin{eqnarray*}
\widehat{\mu}_I^1&=&\arg\min_{X(\alpha)\in\mathcal{C}}\sum_{i=1}^n\delta\left(X\left(A_i\right),X(\alpha)\right)\\
&=&\arg\min_{X(\alpha)\in\mathcal{C}}\sum_{i=1}^nD\left(A_i,\alpha\right)=\arg\min_{X(\alpha)\in\mathcal{C}}C(\alpha)
\end{eqnarray*}
where $D$ is the following distance on $]b,c[$ :
$$D\left(A_i,\alpha\right)=\left|\int_{A_i}^{\alpha}\left\|X^\prime(a)\right\|\mathrm{d}a\right|.$$
In the following, let $\left(A_{(i)}\right)_i$ the ordered parameters. That is
$$A_{(1)}<A_{(2)}<\dots<A_{(n)}.$$
Then, for a given $\alpha\in]b,c[$ such that $A_{(j)}<\alpha<A_{(j+1)}$, we get that
\begin{eqnarray*}
C(\alpha)&=&jD\left(\alpha,A_{(j)}\right)+\sum_{i=1}^{j-1}iD\left(A_{(i)},A_{(i+1)}\right)\\
&+&(n-j)D\left(\alpha,A_{(j+1)}\right)+\sum_{i=j+1}^{n-1}(n-i)D\left(A_{(i)},A_{(i+1)}\right).
\end{eqnarray*}
For the sake of simplicity, let $n=2q+1$. It follows that $\widehat{\mathrm{med}}(A)=A_{(q+1)}$. Moreover, let $\alpha=A_{(j)}$ with $j<q+1$. By symmetry, the case $j>q+1$ will hold. Then, we rewrite $C\left(\alpha\right)$ as
\begin{equation*}
C\left(\alpha\right)=\sum_{i=1}^{j-1}iD\left(A_{(i)},A_{(i+1)}\right)+\sum_{i=j}^{n-1}(n-i)D\left(A_{(i)},A_{(i+1)}\right)
\end{equation*}
and, by introducing $A_{(q+1)}$, we get that
\begin{eqnarray*}
C(\alpha)&=&\sum_{i=1}^{j-1}iD\left(A_{(i)},A_{(i+1)}\right)+\sum_{i=j}^qiD\left(A_{(i)},A_{(i+1)}\right)\\
&+&\sum_{i=j}^q(n-2i)D\left(A_{(i)},A_{(i+1)}\right)+\sum_{i=q+1}^{n-1}(n-i)D\left(A_{(i)},A_{(i+1)}\right).
\end{eqnarray*}
Finally, we notice that
\begin{equation*}
C(\alpha)=C\left(A_{(q+1)}\right)+\sum_{i=j}^q(n-2i)D\left(A_{(i)},A_{(i+1)}\right)>C\left(A_{(q+1)}\right).
\end{equation*}
And the result follows since
$$\widehat{\mu}_I^1=\arg\min_{X(\alpha)\in\mathcal{C}}C(\alpha)=X\left(A_{(q+1)}\right)=X\left(\widehat{\mathrm{med}}(A)\right)=\widehat{f}_\mathrm{SM}.$$
\end{proof}

\bibliographystyle{plainnat}
\bibliography{manif}

\end{document}